%% file: vm-or.tex
\documentclass[11pt]{article} 
\input{macros.tex}

\title{Mechanism Design for Value Maximizers}
\author[$\dagger$]{Christopher A. Wilkens}
\author[$+$]{Ruggiero Cavallo}
\author[$\S$]{Rad Niazadeh}
\author[$\circ$]{Samuel Taggart}
\affil[$\dagger$]{Tremor Technologies,}
\affil[$+$]{Yahoo Research,}
\affil[$\S$]{Department of Computer Science, Stanford University,}
\affil[$\circ$]{Department of Computer Science, Oberlin College.}

\begin{document}
\maketitle

\begin{abstract}

In many settings, money is a tool of exchange with minimal inherent utility --- agents will spend it in a way that maximizes the value of goods received
subject to reasonable constraints, giving only second-order consideration to
the trade-off between value and price. While this perspective is commonly
captured in consumer choice theory, market equilibrium theory, and other
fields, it is markedly absent from the mechanism design literature --- agents
strategizing in a mechanism with money are almost always assumed to incorporate
money as an objective through quasilinear valuations.
\if We introduce a simple model of {\em value maximizers} that captures the
view that money is a constraint, and study mechanisms designed for such agents.
\fi
We study a simple model of {\em value maximizers} that captures online
advertisers and other agents who may view money solely as a constraint, and
study general questions of mechanism design for such agents.   
We show that the feasible and optimal points faced by a mechanism designer
change dramatically from the quasilinear realm and lay a foundation for a
broader study of value maximization in mechanism design.
Along the way, we offer new insight into the generalized second price (GSP)
auction commonly used in Internet advertising.
\if : for value maximizers, GSP is the truthful auction~\citep{aggarwal09}.
Moreover, this implies an axiomatization of GSP --- it is an auction whose
prices are truthful for value maximizers --- that can be applied much more
broadly than the simple model for which GSP was originally designed. In
particular, applying it to arbitrary single-parameter domains recovers the
folklore definition of GSP. \fi
Through the lens of value maximization, GSP metamorphosizes into a truthful
auction, sound in its principles and elegant in its simplicity.
\end{abstract}

\input{introduction3}

\input{prelim}

\input{sp}

\input{sp-revenue}

\input{gsp}

\input{robust}

\input{value-max}

\input{conclusion}

\section*{Acknowledgments.}

We would like to thank Prabhakar Krishnamurthy, the participants of the 2015 Ad Auctions Workshop, and anonymous reviewers for their helpful comments and suggestions.

\bibliographystyle{plainnat} 
\bibliography{vm}

\end{document}

%% file: macros.tex
\usepackage[T1]{fontenc}
\usepackage[utf8]{inputenc}
\usepackage{csquotes}
\usepackage{authblk}
\usepackage[authoryear]{natbib}

\usepackage[margin=1in]{geometry}
\usepackage{ifthen}
\usepackage{graphicx,color,xcolor}

\definecolor{cornellred}{rgb}{0.7, 0.11, 0.11}
\definecolor{dgreen}{rgb}{0.0, 0.5, 0.0}
\definecolor{ballblue}{rgb}{0.13, 0.67, 0.8}
\definecolor{royalblue(web)}{rgb}{0.25, 0.41, 0.88}
\definecolor{bleudefrance}{rgb}{0.19, 0.55, 0.91}
\definecolor{royalazure}{rgb}{0.0, 0.22, 0.66}
\usepackage{url}
\usepackage{hyperref}
\hypersetup{
    colorlinks = true,
    linkcolor=cornellred,
    citecolor=royalazure,
    linkbordercolor = {white}
}

\usepackage{fancybox}
\usepackage{array,float}

\usepackage{amstext,amssymb,amsmath}
\usepackage{amsthm} 
\usepackage{mathrsfs}
\usepackage{dsfont}
\usepackage{pstricks, pst-tree, pst-node}
\usepackage{enumitem}
\usepackage{mathtools}
\usepackage{accents}

\usepackage{romannum}
\AtBeginDocument{\pagenumbering{arabic}}
\usepackage{ifthen}

\usepackage[ruled]{algorithm2e} 

\usepackage{algpseudocode}
\usepackage[capitalize]{cleveref}

\theoremstyle{plain}
\newtheorem{theorem}{Theorem}[section]
\newtheorem{lemma}[theorem]{Lemma}
\newtheorem*{lemma*}{Lemma}

\newtheorem{corollary}[theorem]{Corollary}
\newtheorem{observation}[theorem]{Observation}

\theoremstyle{definition}
\newtheorem{definition}[theorem]{Definition}
\newtheorem{example}[theorem]{Example}
\newtheorem{remark}[theorem]{Remark}

\DeclareMathOperator{\argmax}{argmax}

\usepackage{epsfig}

\def\shortcite#1{\cite{#1}}

%% file: introduction3.tex
\section{Introduction}  \label{sec:intro}

Quasilinear utilities epitomize {\em lex parsimoniae}: they are ubiquitous in
Economics because they are the simplest objective that trades value for money.
However, they presuppose that money is in fact a goal, while many practically significant categories of strategic agents do {\em not} consider money to be a first-order
objective. Rather, they want value, and money is simply treated as a spending
constraint.

Case in point: Amazon.com. Over the past 10 years, the e-commerce giant's profit margins have averaged a mere 2.21\%, much closer to brick-and-mortar stores like Best Buy (1.77\%) and Barnes \& Noble ($-0.18$\%) than e-commerce and Internet companies like eBay (17.18\%) and Google (24.16\%) (profit margin data from \url{ycharts.com}). Why? Amazon's worth is not maximized by making a quick buck in the moment, but by growing to become the 10,000lb gorilla of e-commerce. Reinvesting \$1 will pay out many times over in long-term value, so Amazon's reinvestment is simply limited by incoming revenue. The market has a heuristic for this --- strong revenue (i.e., value) indicates a company's potential for growth, so we can say maximizing Amazon's revenue (subject to profitability) is its goal, not profit (quasilinear utility).

Indeed, many fundamental models in microeconomic theory start by assuming money has no utility. It is standard in consumer choice theory to give each buyer a budget that will be entirely consumed. Similarly, the theory of general equilibrium introduces money as a lubricating tool that almost disappears in equilibrium --- agents sell their goods and buy other goods, consuming their entire income. In these settings, agents are modeled simply by their preference over goods with no benefit from leftover money.

Yet, in mechanism design it is unsatisfactory to assume money is always spent. Auctions and other mechanisms with money exploit quasilinear agents' price-sensitivity to ensure, e.g., that the winner is the bidder with the highest value; sacrificing price-sensitivity renders much of the mechanism design literature meaningless. This leaves us with a hole: what is the {\em lex parsimoniae model} for agents who, while sensitive to price, still see money as mainly a constraint? Our answer lies in a remarkably simple {\em value-maximizing} model:
\vskip -0.4in
\begin{proof}[Definition (Informal)]
\phantom{\qedhere}
A {\em simple value maximizer} strategizes to maximize the value $v_i$ he
obtains \if while maintaining budget balance, i.e., \fi while keeping payment $p_i\leq
v_i$; among outcomes with equal value, a lower price is strictly preferred.
\end{proof}

This definition captures the idea that an agent maximizes her perceived value with only a second-order consideration for saving money. 

An important application of the value maximizing model is when $v_i$ directly measures the maximum the agent is able to pay, e.g., due to budget or Return on Investment (ROI) constraints. 
As we will discuss, advertisers frequently manifest such behavior --- consider an online advertiser who wants to maintain a fixed profit margin. If this advertiser has an average net profit of \$4 from each consumer who clicks on its ad and wants to maintain a fixed profit margin of 50\%, the advertiser would be wiling to pay \$2 per click, i.e., $v_i=2$ for one click, $v_i=4$ for two clicks, and so on. 

Budgets also highlight an important extension to the simple value maximizer
model: agents may have a strict preference over outcomes even though
willingness to pay is the same. This is very naturally motivated when
willingness to pay is determined by {\em ability} to pay (i.e., budget), which
is frequently a fixed amount of cash on hand that provides a uniform binding
constraint regardless of what outcome results. This motivates our full
definition of a value maximizer:

\begin{proof}[Definition (Informal)]
\phantom{\qedhere}
A {\em value maximizer} has a preference over outcomes and a maximum willingness to pay for each outcome. \if (weakly consistent with its preferences). \fi It strategizes to achieve the most preferred outcome without incurring a price above what it is willing to pay.
\end{proof}
For this paper we restrict attention to a refinement of
the model in which the willingness-to-pay ordering does not clash with the
preference ordering.

This strictly generalizes the simple value maximizer. However, in the most
complex settings, the models nearly converge again --- when an agent's
willingness to pay is distinct for every outcome, the model again coincides
with a simple value maximizer. Thus, we will largely focus on simple value
maximizers.

\subsection{Preferences Over Extreme Bundles}

On face, agents who are value maximizers would be expected to exhibit extreme, arguably unnatural behavior.
Note that a value maximizer would choose an outcome with value \$1,000 over an
outcome with value \$999, even if the former cost \$999 and the latter were
free. However, this inherent lexicographic behavior should not precipitate wholesale rejection of the model.
When outcome
values are sufficiently differentiated, we will show that the preferences revealed by an agent
who interacts with an auction will be insufficient to differentiate between a
strict value maximizer (intuitively implausible) and an agent with complex
preferences that limit the cost of marginal gains in extreme cases (quite
plausible). Moreover, we will show that is not a hollow theoretical excuse for
downplaying extreme behavior --- in a standard auction model from Internet
advertising, wide classes
of agent preferences would be indistinguishable from value maximizing
preferences in an auction that is truthful for value maximizers. Thus, truthful
mechanisms for value maximizers are important even though the preference model,
in its simplicity, may be unrealistic in some extreme theoretical scenarios.

\subsection{Motivating settings}

Value maximizing behavior can be seen across a wide variety of settings; we
describe some here. 

\subsubsection{Internet advertisers}

\if Advertisers are quintessential examples of value maximizers, and internet
advertising will be a central example that we return to at various points in
the paper. \fi
Internet advertisers are quintessential value maximizers. They can generally be
categorized as either {\em brand} or {\em performance}, both of which are
well-captured by the {\em simple value maximizer} model, optimizing value given
different kinds of spending constraints.

\paragraph{Preferences of Brand Advertisers} Brand advertisers aim for
long-term growth and awareness.
They come to a marketplace with explicit objectives, and typically have a
mandate to meet a specific business goal---showing impressions to an audience,
generating clicks, or maximizing revenue---driven by long-term considerations
instead of immediate profit. Thus cost, while important, merely enters their
preferences as a constraint; brand advertisers will have a budget and limits on
what they are willing to pay on average for an impression/click/conversion, but
otherwise directly optimize for their mandate.

\paragraph{Preferences of Performance Advertisers} Performance advertisers, on
the other hand, optimize the immediate tradeoff between value---measured as
sales, sign-ups, or other so-called {\em conversions} generated directly from
their ads---and cost.  Return on investment (ROI) has been the standard metric
for measuring this tradeoff across all types of advertising for decades. ROI
measures the ratio of the profit obtained (``return'') to the cost or price
paid (``investment''), i.e., the density of profit in cost:
\[\texttt{ROI}=\frac{\texttt{Value-Price}}{\texttt{Price}}=\frac{\texttt{Profit}}{\texttt{Price}}\enspace\]
Being a density metric, unconstrained maximization of ROI is not
sensible, mainly because unconstrained maximization of ROI would likely cause an
advertiser to buy only the single cheapest impression or click available;
instead, advertisers come with an ROI constraint and optimize (e.g., maximize
revenue and/or clicks).

For example, the following story plays out regularly in Yahoo's ad
marketplaces:

\begin{enumerate}
\item Advertiser X designates a small budget for testing a Yahoo advertising
product.
\item Advertiser X measures ROI --- if X is happy with its ROI, it increases
its budget hoping to {\em maintain the same ROI}; if it is unhappy, it
withdraws.
\end{enumerate}

This behavior is also reflected in the standard industry tools: Google's
AdWords campaign management tool buys as much advertising as possible while
maintaining a target average cost-per-click (CPC) and budget (an average CPC constraint
corresponds to an ROI constraint, while a marginal CPC constraint would be
appropriate for profit maximization). Yahoo's display ad products also offer
options with ``CPC goals'' and ``CPA goals'' that buy impressions while aiming
to achieve a target average cost-per-click or cost-per-conversion, respectively.

Performance advertisers' ROI-centric behavior {\em does not} follow the
standard quasilinear intuition. One way to resolve this conflict is to conclude
that the intuition is wrong, perhaps since companies often care about revenue
and margins as much as profits. Another resolution is to conclude that ROI is a
deeply-ingrained heuristic for maximizing the effectiveness of a fixed
resource, either managing an advertising
budget~\citep{borgs07,kitts04,szymanski08,zhou08,auerbach08} among auctions or
between platforms (deciding how to bid between Yahoo and Google), or even
managing resources within the advertiser's organization (allocating a budget
between marketing and engineering). But regardless, performance advertisers
behavior is rarely profit-maximizing.

\subsubsection{Budgets as Targets}

For many agents, the budget represents not simply an upper bound but also a target. For them, budgets are the primary lever used to control their spending, in sharp contrast with auction theory, where it is typically assumed that the most significant strategic parameter is the bid. Indeed, in Internet advertising, many advertisers do not come with a bid but instead come with only a budget. Intermediary algorithms then manage the advertiser's bid to spend the budget fully. Such agents want the platform to maximize value subject to the specified budget, otherwise this tuning process is nonsensical.

More broadly, in large organizations, ``budget'' is typically a limited resource that must be carefully provisioned among the many groups that need it. As a result, each group only gets what it needs, and it is expected that budgets will bind. (This leads to a common and perverse phenomenon wherein groups that underspend get a smaller budget during the next budgeting cycle.) Again, such groups are incentivized to spend as much as they can within business constraints.

\subsubsection{Agency}

Any agent working on behalf of a principal will often be asked to and incentivized to maximize value. For example, advertising agencies are commonly paid by commission. Consequently, an agency's utility is maximized by spending as much of the client's money as possible without violating constraints set forth by the client.

A more mundane example is an employee ordering coffee for a meeting. The
meeting's attendees will focus on the quality of the coffee with little regard
for its cost, if they even know it. As a result, the employee ordering
coffee---whose performance will be evaluated first on the quality of the
coffee---will optimize by getting the best coffee available within an implicit
or explicit budget.

The general problem is due to the nature of agency --- enlisting an agent on one's behalf typically requires fixing a budget and other constraints. Once these constraints are set forth, the agent is primarily evaluated on maximizing value.

\subsubsection{Companies}

Almost paradoxically, companies commonly exhibit value-maximizing behavior
despite the fact that their {\em raison d'\^etre} is ostensibly to make profit.
One explanation is that while profit captures short-term performance, the
long-term prospects of a business are often divined from its revenue. A company
that maximizes revenue subject to maintaining profitability (revenue exceeds
cost) is a textbook example of a simple value maximizer.

While using revenue as a bellwether for long-term performance is a heuristic, it is plausibly grounded when long-term value from investment far exceeds immediate payoffs. Quasilinear utilities
assume free and unlimited borrowing, which is unrealistic in practice. Rather,
companies are generally expected to remain instantaneously profitable, and so their present investment in long-term value is capped by their immediate profits.
For example, the long-term value of a new customer can easily be an order of
magnitude higher than the profit attributed to said customer in its first
quarter or year. As long as a company's long term value exceeds what it is able
to pay (either constrained by revenue or by practical limits on borrowing
against future profits), its optimal strategy is to bring in as many customers
as possible.

Amazon exemplifies how a focus on revenue optimizes for long-term value as
noted earlier. Yahoo's recent history (prior to being acquired) is another
example: investors' headline goal and demand was revenue growth as an
indicator that the company had long-term value.

\subsubsection{Consumers}

For the consumer, value-maximizing behavior can be understood as comparing the price of an item to an intrinsic value benchmark rather than considering the utility of leftover money. For example, one author commonly employs this strategy to buy strawberries. He buys up to 4 pounds of strawberries each week depending on price according to a value maximizing strategy --- roughly, if he can get 3 pounds for \$8 or less, he buys 3 pounds; if 3 pounds is too expensive, he will buy 2 pounds if the price is at most \$6; finally, if he can't buy more, he buys 1 pound as long as the price is at most \$5. In particular, the author's heuristic behavior is based on (learned) price thresholds for each quantity; it is {\em not} based on the difference between value and price, as a quasilinear model would require.

Neither is value maximizing behavior restricted to little purchases. Consider housing purchases. A house buyer typically starts with limits on what he is willing or able to pay, then finds the best house that satisfies these constraints. For example, a buyer might have a budget of \$400,000 for the perfect house, but a limit of \$300,000 if the house only has 3 bedrooms. Subject to these constraints, the buyer will find the house on the market with the highest value.


\subsection{Our Results}


We address fundamental mechanism design questions for value maximizers. First,
what are truthful mechanisms? For the single-parameter setting, we provide a
characterization (Section \ref{SEC:SINGLEPARAM}).  We then use this
characterization to derive the revenue-optimal mechanism for the simple setting
of selling $k$ identical items to a single buyer. We show that the
revenue-optimal ex-post truthful mechanism heavily exploits second-degree price
discrimination by limiting the number of goods sold, in sharp contrast to
Myerson's optimal mechanism for quasilinear buyers that only requires a reserve
price and therefore sells all $k$ items or nothing.
%
In Sections \ref{sec:gsp} and \ref{sec:robust}, we then use our characterization of truthful mechanisms to explain the centrality of the GSP mechanism
for advertising auctions. We demonstrate that GSP is the truthful mechanism for ad
auctions with value maximizers.

In Section \ref{sec:value-max}, we study mechanisms when values over outcomes are fully general, showing that any mechanism from a family of transformed greedy
mechanisms is always truthful even for unrestricted valuation structures. This
family is analogous to the Vickrey-Clarke-Groves (VCG) auction and affine
maximizers in quasilinear settings. Interestingly, while affine maximizers in
quasilinear settings are limited to scaling bidders' values, our greedy
mechanisms are truthful for an arbitrary monotone transformation of bidders'
values. Unfortunately but unsurprisingly, while these mechanisms are powerful
because of their generality, since they are greedy we show that they can only
guarantee a $\Theta(n)$-approximation to the optimal sum of agents' values.


\subsection{Related Work}

The idea that agents maximize value is common. For example, the traditional studies of consumer choice and market equilibria begin with agents who pick the most-preferred bundle they can afford based on their endowment~\citep{mas-collel95}.

In the context of mechanism design, a limited literature studies non-quasilinear bidders. 
Concurrent work by \shortcite{fadaei16} introduces a similar model and discusses a plethora of similar and interesting motivating examples, then studies revenue maximization. Our model differs notably because we say bidders prefer lower prices when value is held constant, while Fadaei and Bichler imply bidders strictly prefer to spend their money. This difference fundamentally changes the mechanism design problem, since without a weak preference for retaining money, Fadaei and Bichler would, e.g., prescribe a first price auction when selling a single indivisible item while our model would require a second price auction. 

In the context of advertising auctions, \shortcite{aggarwal09} design sponsored search (slot) auctions for advertisers whose preferences add constraints on top of a quasilinear utility. While their general models are technically incomparable to ours, they cover our value maximizers in the special case of sponsored search. \shortcite{alaei11} show how Walrasian equilibria in non-quasilinear unit-demand settings can be leveraged to build ad auctions where estimation errors break quasilinearity.

A few papers study general truthful mechanisms for non-quasilinear preferences~\citep{adachi13,morimoto15}. These papers develop axiomatic characterizations of a VCG-analog for multi-unit and unit-demand settings.

%% file: prelim.tex
\section{Model and Preliminaries}  \label{sec:prelim}

A mechanism $M$ picks an outcome $o\in\mathcal O$ and prices $p_i\in\Re$ for
each agent $i\in[n]$. Each agent has a total preference order $\preceq_i^o$
over $\mathcal O$ and a willingness-to-pay $v_i:\mathcal O\rightarrow\Re$ for
each outcome. 
%

%
Together, $\preceq_i^o$ and $v_i$ define a total preference order
$\preceq_i^{(o,p)}$ over outcome-payment pairs $(o, p)$ as follows:
\begin{definition}
A {\em value maximizer} prefers bundles $(o,p)$ according to $\preceq_i^o$ as long as
$p\leq v_i(o)$, with ties broken in favor of a lower price. When $p>v_i(o)$ a
value maximizer prefers a lower price, with ties broken in favor of
$\preceq_i^o$.
\end{definition}

We will largely focus on simple value maximizers, whose preferences are
entirely defined by $v_i$, i.e., with $o_1\succeq_i^oo_2$ if and only if
$v_i(o_1)\geq v_i(o_2)$.


\begin{definition}
A {\em simple value maximizer} prefers bundles $(o,p)$ with higher
$v_i(o)$ as long as $p\leq v_i(o)$, with ties broken in favor of a lower
price. When $p>v_i(o)$ a lower price is preferred, with ties broken in
favor of a higher value. Formally, for any two outcome-payment pairs $(o_1,p_1)$ and $(o_2,p_2)$,
$(o_1,p_1)\succ_i^{(o,p)}(o_2,p_2)$ if and only if any of the following
hold:
\begin{itemize}
\item $v_i(o_1)\geq p_1 \;\wedge\; v_i(o_1)>v_i(o_2)$
\item $v_i(o_1)\geq p_1 \;\wedge\; p_1< p_2 \;\wedge\; v_i(o_1)=v_i(o_2)$
\item $v_i(o_2)<p_2 \;\wedge\; p_1< p_2$
\item $v_i(o_2)<p_2 \;\wedge\; p_1=p_2 \;\wedge\; v_i(o_1)>v_i(o_2)$
\end{itemize}
\end{definition}

For comparison, a quasilinear bidder always prefers the bundle that maximizes the difference between value and price.

\begin{definition}
A {\em budgeted value maximizer} has a preference ordering $o_1\succeq_i^oo_2$ over outcomes but the same willingness to pay $v_i(o)=B$ for all of them.
\end{definition}

\begin{definition}
A bidder with {\em quasilinear} preferences prefers bundles with higher $v_i(o)-p$, i.e.,
$v_i(o_1)-p_1\geq v_i(o_2)-p_2 \Leftrightarrow (o_1,p_1)\succeq_i^{(o,p)}(o_2,p_2)$.
\end{definition}

\paragraph{Mechanisms and Truthfulness.} A direct revelation mechanism $M$ takes reports of bidders' preferences as input. We will work with simple value maximizers, whose preferences are entirely encoded in $v_i$. 
Given reported values, also known as \emph{bids} and denoted by $\mathbf{b}$, a mechanism chooses an outcome $f(\mathbf{b})$ and payments $p_i(\mathbf{b})$. We will often suppress $f(\cdot)$ when it is clear from context, e.g., writing $v_i(\mathbf{b})$ instead of $v_i(f(\mathbf{b}))$.
A mechanism is {\em truthful} if reporting a bidder's true private information (e.g., bidding $b_i=v_i$ for simple value maximizers) is a dominant strategy:

\begin{definition}
A mechanism $\mathcal M=(x,\mathbf{p})$ is {\em dominant strategy incentive compatible} (DSIC) if, for any $i$ and bids $b_{-i}$, we have 
\[ (f(b_{-i},v_i),p_i(b_{-i},v_i)) \succeq_i (f(b_{-i},b_i),p_i(b_{-i},b_i)) \]
for all value functions $v_i$ and deviations $b_i$.
\end{definition}

In this paper, 
we will have trouble in continuous bid spaces when truthful bidding would induce a tie; to solve this we use a slight weakening of DSIC that allows arbitrary behavior on a set of bids that ``never occur'':

\begin{definition}
\label{def:dsic-ae}
A mechanism $\mathcal M=(x,\mathbf{p})$ is {\em dominant strategy incentive compatible almost everywhere} (DSIC-AE) if, for any $i$ and bids $b_{-i}$, we have 
\[ (f(b_{-i},v_i),p_i(b_{-i},v_i)) \succeq_i (f(b_{-i},b_i),p_i(b_{-i},b_i)) \]
for all value functions $v_i$ and deviations $b_i$, except a set of $v_i$ with measure zero according to the Lebesgue measure over $v_i$ (unless specified otherwise).
\end{definition}
We loosely use the term ``truthful'' to describe either DSIC or DSIC-AE.

\begin{remark}
Note that DSIC is the same as DSIC-AE for a finite bid space. Our DSIC-AE mechanisms will be DSIC when applied to a well-ordered bid space.\end{remark}


%% file: sp.tex
\section{Characterizing Mechanisms for Single-Parameter Domains}
\label{SEC:SINGLEPARAM}

To begin, we study truthful mechanisms whose private information can be represented by a single parameter, specifically where the valuation function can be factored as $v_i(o)=v_i\cdot x_i(o)$. In quasilinear settings it is known that any monotone allocation rule can be made incentive compatible~\citep{myerson81, archer01} when $x_i(o)$ is public information and $v_i$ is private. Two natural single-parameter models come to mind when studying value maximizers: {\em simple value maximizers} instantiate our canonical model, preferring higher-value outcomes as long as the price does not exceed the value; and {\em budgeted value maximizers} who prefer higher-valued outcomes as long as the payment does not exceed the budget. In both cases, we will see that truthful mechanisms have very simple prices: the price is the minimum value required to maintain the same allocation $x$.

\vspace{1mm}
\paragraph{Simple Value Maximizers}


\begin{theorem}
\label{THM:SPDSIC}
For simple value maximizers in a single parameter domain, a mechanism is truthful (DSIC-AE) if and only if the allocation is monotone and the price is the minimum value required to get the same allocation:
\begin{itemize}
\item {\em (monotonicity)} $[x_i(z, v_{-i})> x_i(v)]\Leftrightarrow[z>v_i]$ almost everywhere
\item {\em (pricing)} $p_i(v)=\inf_{z|x_i(z, v_{-i})=x_i(v)}zx_i(v)$
\end{itemize}
\end{theorem}
\begin{corollary}
\label{COR:CRITICALPAY}
For any monotone $x_i$, the mechanism that charges $p_i(v)=\inf_{z|x_i(z, v_{-i})=x_i(v)}zx_i(v)$ is DSIC-AE.
\end{corollary}
This corollary follows because the $\inf$ only differs from the $\min$ where $x$ is discontinuous, and a monotone function is continuous almost everywhere (discontinuities are countable and therefore have measure zero).

\proof{Proof of Theorem~\ref{THM:SPDSIC}.}

Fix other bidders' values $v_{-i}$ and drop them for clarity.

{\em Necessity.} Fix a type $v_i$ and define
\[\mathcal{Z}^<=\{z|x_i(z)<x_i(v_i)\}\quad\mbox{and}\quad \mathcal{Z}^==\{z|x_i(z)=x_i(v_i)\}\]
to be the types that get a strictly smaller allocation than $v_i$ and an equal allocation respectively. A value maximizer with single parameter type $v_i$ will choose a bid $b_i$ that maximizes $x_i(b_i)$ subject to $p_i(b_i)\leq v_i$. Thus, for any $v_i$, $p_i(v_i)$ must be high enough that bidders who get a smaller allocation do not want to lie, i.e.
\[p_i(v_i)\geq \underline p_i\triangleq\sup \mathcal{Z}^<\enspace.\]
Similarly, all types who exactly get $x_i(v_i)$ must pay the same price, and it must be that any bidder who gets $x_i(v_i)$ must be willing to pay for it, so
\[p_i(v_i)\leq\overline p_i=\inf \mathcal{Z}^=\enspace.\]
Since $\sup\mathcal{Z}^<\geq\inf\mathcal{Z}^=$, it must be that $\underline p_i\geq\overline p_i$, and therefore to satisfy $\underline p_i\leq p_i(v_i)\leq \overline p_i$ a mechanism will be truthful if and only if $\underline p_i=\overline p_i=p_i(v_i)$ --- this happens if and only if $x$ is nondecreasing and is the price defined in the theorem.

{\em Sufficiency.} Note that for DSIC-AE we only need to prove sufficiency almost everywhere. We must consider three deviations: (1) a bidder underbids for a smaller allocation, (2) a bidder deviates for the same allocation at a lower price, and (3) a bidder overbids for a larger allocation. Note that deviation (1) will not happen because a value maximizer will always want the larger allocation (since $p_i(v_i)\leq v_i$ by definition of $p_i$), and that  (2) cannot happen because all types who get the same allocation pay the same price.

It remains to show that bidders almost never have an incentive to raise a bid (3). Deviation (3) will only be beneficial if there exists a type $z>v_i$ such that $x_i(z)>x_i(v_i)$ but $p_i(z)<v_i$. By definition of $p_i$, there exists a type $\underline b_i(z)\geq v_i$ such that
\[p_i(z)=\underline b_i(z)\enspace.\]
As long as $\underline b_i(z)>v_i$ this deviation cannot be beneficial. Unfortunately, when $b_i(z)=v_i$ then deviation (3) will be beneficial; however, this can only happen at a discontinuity in $x$. Since $x$ is monotone, discontinuities can only happen rarely (the set of discontinuities is countable and therefore has measure 0) and we get truthfulness in our almost-everywhere sense (DSIC-AE).\qed
\endproof

\vspace{1mm}
{\em Budgeted Value Maximizers}

Budgets are another natural setting where value maximizers have a single piece of private information to reveal. It is often the case that an agent's preference ordering $\preceq_i^o$ over outcomes is public knowledge (say, advertisers obviously prefer to be ranked higher on the page) and an agent has a fixed budget irrespective of the outcome. In this case, the only information that a bidder must reveal is the maximum amount it is willing to pay, $v_i$. We find that a nearly identical result holds in this setting:

\begin{theorem}\label{thm:vm-budget-char}
For value maximizers with a single budget $v_i=B_i$, a mechanism on a countable outcome space is truthful (DSIC-AE) if and only if it is monotone and the price is the minimum budget required to get the same allocation:
\begin{itemize}
\item {\em (monotonicity)} for any $v_{-i}$, $f(z, v_{-i})\succ_i^o f(v)$ if and only if $z>v_i$ almost everywhere over $v_i$
\item {\em (pricing)} $p_i(v)=\inf\{z|f(z, v_{-i})=f(v)\}$
\end{itemize}
\end{theorem}
We also observe an analogous corollary:
\begin{corollary}
For any monotone $x_i$, the mechanism that charges $p_i(v)=\inf\{z|f(z, v_{-i})=f(v)\}$ is DSIC-AE.
\end{corollary}
The proof of Theorem~\ref{thm:vm-budget-char} is nearly identical to the proof of Theorem~\ref{THM:SPDSIC} and is omitted.

%% file: sp-revenue.tex
\section{Maximizing Revenue from a Single Agent}\label{sec:sp-revenue}

In this section, we use the the characterization just derived for a
single-dimensional simple value maximizer to study mechanisms that maximize
revenue. We suppose the type is a {\em value} drawn from a distribution $F$ and
maximize expected revenue among deterministic, DSIC mechanisms for selling $k$
identical items. Any such mechanism maps the agent's value $v$ to a
deterministic allocation $x\in \{0,1,\ldots,k\}$. The agent gets value for
allocation $v\cdot x(v)$ (i.e., its value is additive across identical items),
and is charged a deterministic payment $p(v)$. (Note that the same problem for
randomized mechanisms is ill-posed without a model for how value maximizers
evaluate lotteries over allocations.)


We start with an example where revenue-optimality entails {\em second-degree
price discrimination}, i.e., the agent is offered the option to buy any number
of items from $0$ to $k$, and a positive measure of values will take each of
the $k+1$ different offers. This stands in contrast to the optimal mechanism
for a quasilinear agent, in which the menu consists of $0$ items and $k$ items,
with nothing in between. We then generalize and provide a revenue-optimal
single-agent mechanism that can be computed in polynomial time via dynamic
programming, and show that across all possible inputs, the price discrimination
it applies is the necessary rule rather than the exception.

\subsection{Example: Two Items, Uniform Agent}
\label{sec:discexample}

Consider selling a supply of two items to an agent with value $v$ drawn
uniformly from $[0,1]$. If the agent was quasilinear, the results of
\cite{myerson81} imply that the optimal mechanism would be to sell the bundle
of both items at price $1$ if $v\geq 1/2$, and to sell nothing otherwise. We
will show that if the agent is a simple value maximizer, the optimal mechanism
will sell a single item to a positive measure of values, in addition to selling
both items and no items at all with positive probability.

To derive the optimal mechanism, we use the characterization in
Section~\ref{SEC:SINGLEPARAM} for single-parameter agents. Namely, we optimize over all
deterministic, monotone allocation rules, which are given by two values: the
value $t_1$ where the allocation steps up from $0$ to $1$, and the value $t_2$
where the allocation steps up from $1$ to $2$. Individual rationality implies
that the truthful payments for an agent with value $v<t_1$ is 0.
Corollary~\ref{COR:CRITICALPAY} implies that the truthful payment for an agent
with value $v\in [t_1,t_2)$ is $t_1$, and for an agent with value $v\geq t_2$
is $2\cdot t_2$. The revenue of our mechanism is therefore $t_1(t_2-t_1)+2
t_2(1-t_2)$, where the first term is revenue from the event that the agent
receives a single item, and the second term from the event where the agent
receives two items. All that remains is to optimize this function over $0\leq
t_1\leq t_2\leq 1$, which yields an optimal solution of $t_1= 2/7$, and $t_2 =
4/7$. Hence, the optimal mechanism allocates exactly one item with probability
$2/7$. This is second-degree price discrimination.


\subsection{Optimal Single-Agent Mechanisms}

We now give a polynomial-time algorithm to compute the optimal mechanism for
selling $k$ items to a simple value maximizing agent. Note that for a
deterministic mechanism, Theorem~\ref{THM:SPDSIC} implies that we need only
choose the thresholds $t_j$ at which the allocation level increases from $j-1$
to $j$ for all $j$ from $1$ to $k$. For simplicity of exposition, we assume the
value distribution $F$ is finitely supported on a set $v_1,\ldots,v_n$, with
probability mass function given by $f(\cdot)$. Label this set such that $v_1 <
v_2 < \ldots < v_n$. Extending our algorithm to continuous distributions is
straightforward for most standard models for distributional access.

Corollary~\ref{COR:CRITICALPAY} implies that the expected revenue of a
mechanism with thresholds $t_1\leq t_2\leq\ldots\leq t_k$ is given by:

\begin{equation}
\label{eq:optobjective}
\sum_{\ell=1}^k \ell\cdot t_\ell(F(t_{\ell+1})-F(t_\ell)),
\end{equation}
where $F(v_i)$ is given by $\sum_{j=1}^{i-1} f(v_i)=\text{Pr}[v<v_i]$, and
where $F(v_{n+1})$ is defined to be $1$.

Define $\text{OPT}(i, j)$ to be the maximum attainable revenue from only
selling to values at least $v_j$, and only selling $i$ or more items in the
event of a sale. Then the value of (\ref{eq:optobjective}) with the optimal
thresholds is equal to $\text{OPT}(1,1)$. Moreover, we have the following
recurrence relation:

\begin{equation}
\label{eq:recurrence}
\text{OPT}(i,j)=\max_{\ell\geq j}\left[iv_j(F(v_\ell)-F(v_j))+OPT(i+1,\ell)\right]
\end{equation}

This yields a straightforward two-dimensional dynamic program. For $i=k$, we
have $\text{OPT}(i,j)=kv_j(1-F(v_j))$ for all $j\leq n$. For $j=n+1$, we have
$OPT(i,j)=0$ for all $i$. The remaining values of $OPT(i,j)$ can be computed by
filling a memoization table starting with high values of $i$ and $j$ and moving
to low values.

\subsection{The Ubiquity of Price Discrimination}

We now show that in a sense, the second-degree price discrimination observed
above is unavoidable. In particular, it will occur for \emph{every} continuous,
fully-supported value distribution. Formally:

\begin{theorem}
\label{THM:REV}
For every continuous, fully-supported value distribution $F$, and for every number of items $i$ between $0$ and $k$, there is a positive measure of values who will purchase exactly $i$ items under the revenue-optimal mechanism to sell to a simple value maximizer with value distribution $F$.
\end{theorem}

\begin{proof}{Proof of Theorem~\ref{THM:REV}}
Theorem~\ref{THM:SPDSIC} implies that every deterministic truthful mechanism is characterized completely by its choice of thresholds $t_j$ at which the allocation level increases from $j-1$ to $j$ for all $j$ from $1$ to $k$. Moreover, we have that the revenue of such a mechanism is given by 
\begin{equation}
\label{eq:optobjectiveagain}
\sum_{\ell=1}^k \ell\cdot t_\ell(F(t_{\ell+1})-F(t_\ell)).
\end{equation}

We consider a mechanism which does not sell $i$ items with positive probability, and show how to improve the mechanism's revenue --- hence, such a mechanism cannot be optimal. A similar argument implies that no mechanism which sells with probability $1$ is optimal.

A mechanism which for some $i\in\{1,\ldots,k\}$ does not sell exactly $i$ items with positive probability must have $t_i=t_{i+1}$. Choose the lowest $i$ for which this is the case. Then $t_i-t_{i-1}>0$. We will show that such a mechanism's revenue can be improved by slightly decreasing $t_i$ to $t_i'=t_i-\epsilon$ for some sufficiently small $\epsilon>0$. This decreases the measure of values purchasing exactly $i-1$ items by $F(t_i)-F(t_i-\epsilon)$, which decreases the revenue contribution from this event by $(i-1)t_{i-1}(F(t_i)-F(t_i-\epsilon))$. Meanwhile, decreasing $t_i$ causes a $F_i(t_i)-F(t_i-\epsilon)$ measure of values to purchase exactly $i$ items instead, yielding a revenue increase of $i(t_i-\epsilon)(F_i(t_i)-F(t_i-\epsilon))$. The gain from decreasing $t_i$ outweighs the loss as long as $(i-1)t_{i-1}<i(t_i-\epsilon)$, which holds as long as $\epsilon$ is sufficiently small.
\qed
\end{proof}

%% file: gsp.tex
\section{A Generalizable Definition of GSP}  \label{sec:gsp}

A significant byproduct of our existential theory from Section~\ref{SEC:SINGLEPARAM} is that we gain a new perspective on the generalized second price (GSP) auction commonly used in Internet advertising.

When Google launched its pay-per-click AdWords auction in February 2002, it
aimed to alleviate stability issues in contemporary first price models by
incorporating ideas from a second price auction. Google quickly realized that
its implementation did not generalize a second price auction according to the
standard theory---the proper generalization would be the Vickrey-Clarke-Groves
(VCG) mechanism---but the GSP auction they invented was performing well enough
that the potential gains of switching to VCG were not worth the risk and
disruption~\citep{varian14}. (Overture.com, formerly GoTo.com and later bought
by Yahoo, introduced the first sponsored search auctions in 1998 and was still
the dominant player in 2002. Their auction had pay-your-bid pricing that has
since been dubbed the generalized first price auction, or GFP.)


While GSP has thrived in the industry, the research community has struggled to rationalize it. Initial analyses showed that revenue and welfare in reasonable equilibria are no worse than VCG~\citep{varian07,eos07}. Later studies offered more nuanced analyses --- a more recent justification falls out of GSP's simplicity: 
GSP's equilibrium guarantees are immune to certain errors in click modeling~\citep{milgrom10,dutting16}. Unfortunately, none of these results would guide us to invent GSP if it weren't already implemented in practice.


Our work suggests a a simpler and more powerful rationalization for GSP that we will explore in this section: {\em GSP is truthful for value maximizers}.

\subsection{ROI, Value Maximizers, and GSP}

We previously argued that advertisers' objectives are well-modeled by value maximization subject to a ROI constraint. We also hinted that this was equivalent to our simple value maximizer model, which we now formalize: 
\begin{lemma}\label{lem:svm-roi}
If a deterministic auction is truthful for simple value maximizers, it will be a dominant strategy for a value maximizer with value $v_i$ and ROI constraint $\gamma$ to report what it is willing to pay, $v_i'=w_i=\frac{v_i}{1+\gamma}$.
\end{lemma}
\begin{proof}{Proof of Lemma~\ref{lem:svm-roi}.}
Both bidders have the same preferences over bundles that are preferred to non-participation, so a dominant strategy for a simple value maximizer with value $w_i$ will also be a dominant strategy for a value maximizer with value $v_i$ and ROI constraint $\gamma$.
\qed
\end{proof}
Note that we no-longer expect bidders to report $v_i$ and instead expect they will report what they are willing to pay, $w_i$. If we desired a true direct revelation mechanism we could ask bidders to report $\gamma$ as well, but in practice it is more natural to report what one is willing to pay than what something is truly worth.

With this lemma in hand, it remains to show that GSP is truthful for simple value maximizers. Surprisingly, this is nearly trivial, and was observed by \cite{aggarwal09} in the context of a general preference model. The proof via \cite{aggarwal09} observes that value maximizers with ROI constraints can be mapped to their maximum price bidder model. We sketch a proof from first principles for intuition, and generalize it to value maximizers with an ROI constraint. A rigorous generalization of this proof can be found in Theorem~\ref{THM:SPDSIC}.
\begin{theorem}[Application of \cite{aggarwal09}]
GSP is truthful (effectively DSIC-AE, see Definition~\ref{def:dsic-ae}) for value maximizers with ROI constraints, that is, it is a dominant strategy for a bidder to report the maximum price she is willing to pay. 
\end{theorem}
\begin{proof}{Sketch of the proof.} Note that the GSP price $p_i$ is the minimum bid $i$ could have submitted while maintaining the same rank in the auction.

First, assume that bidders are simple value maximizers. Fix a bidder $i$ and bids $b_{-i}$. By the taxation principle, we can think about bidder $i$ choosing among the slots at prices $p_{i,j}$. Let $V_{i,j}\subseteq\Re$ denote the set of value reports for which $i$ wins slot $j$ under a particular auction.

For a pricing to be truthful for a value maximizer, it should be that $i$ wins a particular slot $j$ if and only if (a) it is willing to pay $p_{i,j}$ for slot $j$ and (b) is not willing to pay $p_{i,j-1}$ for the slot immediately above it. This implies that (a) $p_j\leq\inf V_{i,j}$ and (b) $p_{j-1}\geq\sup V_{i,j-1}$ and thus
\[\sup V_{i,j+1}\leq p_{i,j}\leq\inf V_{i,j}\]
that is, $p_{i,j}$ is the threshold value (bid) at which $i$ moves from slot $j+1$ to slot $j$. This is precisely the GSP price. Finally, with the help of Lemma~\ref{lem:svm-roi}, we can extend above to value maximizers with an ROI constraint.
\qed
\end{proof}

\subsection{Generalizing GSP}

A significant consequence of this new view of GSP is that we discover a principled way to generalize it.  We first remark that GSP is often described by the following folklore definition (this definition completely specifies GSP pricing in the standard position auction model): 
\begin{proof}[Definition (Folklore)]
\phantom{\qedhere}
GSP is the auction that maximizes expected bidder value and charges the minimum bid required to keep the same allocation.
\end{proof}
Until now, this definition was merely a pleasant-sounding heuristic for general auction settings; however, our work gives this definition theoretical teeth:
\begin{observation}
The folklore definition of GSP is equivalent to the truthful auction for value maximizers defined in Theorem~\ref{THM:SPDSIC}.
\end{observation}

This inspires the following generalized definition of GSP:
\begin{proof}[Definition (Proposed)]
\phantom{\qedhere}
\label{prop:ggsp}
A generalized second price auction (GSP) maximizes expected bidder value and charges prices that make the auction truthful for simple value maximizers.
\end{proof}
This definition is important because it allows us to extend GSP to new contexts in a principled way. Doing so is especially important problem because ads are becoming progressively more complex, and the algorithms required to optimize their placement and features are following suit. A small existing literature has struggled to generalize GSP in the context of quasilinear bidders, generally finding equilibria may not exist and that even when they do, they may not be efficient~\citep{deng10,cavallo14,bachrach14}. Our characterization suggests that these negative results may be because GSP's strength was misunderstood.


%% file: robust.tex
\section{Robustness Value Maximizing Model in Internet Advertising}  \label{sec:robust}

We now show that our minimal theoretical model is robust in a practical sense. Value maximization is intuitively an extreme behavior, so it would be unsurprising if small modeling changes dramatically changed our results. However, we find that the opposite happens in the presence of ROI constraints: a sufficient ROI constraint generally makes nearly any advertiser look like a value maximizer. An example is instructive:

\begin{example}
Suppose that a bidder has an ROI constraint of 1 and has choices between an outcome $o$ with value $v_ix_i(o)=1$ and an outcome $o'$ with value $v_ix_i(o')=10$. Assuming $p\geq0$, a profit-maximizing bidder will prefer $o'$ to $o$ at any price $p'\leq\$9$; however, since bidder $i$ has an ROI constraint of 1, she will only consider outcome $o'$ when $p'\leq\$5$. As a result, any time $o'$ is cheap enough for her to consider it, she will always prefer it to outcome $o$ --- in this example, {\em bidder $i$ is effectively a simple value maximizer with value $v_i'=\frac12v_i$.}
\end{example}

More generally, when outcomes in an auction lead to dramatically different allocations, a bidder must face a very high price before a lesser outcome (lower value) would become preferable. In such cases, a mild ROI constraint---which caps the price an advertiser is willing to pay---will push a bidder towards simple value maximizing behavior.
%
\subsection{Theoretical Robustness}
To formalize this, we define a broad class of preference relations that captures everything from quasilinear to value-maximizing behaviors:
\begin{definition}
A preference relation $\prec_i$ is {\em super-quasilinear} if an agent who prefers a higher-value option under quasi-linear preferences also always prefers this option under $\prec_i$ when both are preferred to non-participation, i.e.,
\begin{gather*}
v_ix_i(o)\geq v_ix_i(o') \mbox{ and } (v_i-p_i)x_i(o)\geq (v_i-p_i')x_i(o')\geq0 \quad\Rightarrow \quad(o,p_i)\succeq_i(o',p_i')\enspace.
\end{gather*}
and prefers non-participation to any option with $p_i>v_i$.
\end{definition}
This class includes traditional quasilinear bidders and value maximizers, as well as families of preferences that interpolate between them, such as profit-maximizing preferences with an ROI constraint $\gamma$ (when $v_i$ represents the maximum price that satisfies $i$'s ROI constraint).

For super-quasilinear bidders we derive a condition under which an ROI constraint makes all preferences look value-maximizing. Our condition simply says that the $x_i$'s generated by the mechanism are sufficiently different that the ROI constraint precludes any scenarios in which value-maximizing and quasilinear preferences would result in a different best response for $i$. 

\begin{theorem}\label{thm:robust}
When bidders have super-quasi-linear preferences $\prec_i$, an ROI constraint $\gamma$, and
\[v_ix_i(o)\frac{\gamma}{\gamma+1}\geq v_ix_i(o')\mbox{ for all }v_i,o,o'\mbox{ where }x_i(o)>x_i(o')\enspace,\]
no bidder wishes to lie in an auction that is truthful, is individually rational for value maximizers, and has no positive transfers (no positive transfers means that the auctioneer never pays the bidders; individual rationality says that every bidder is at least as happy as if she had not participated at all.)
\end{theorem}
\begin{proof}{Proof of Theorem~\ref{thm:robust}.}
	Let $o$ and $p_i$ be the outcome and price that the bidder sees if it reports truthfully. Let $o'$ and $p_i'$ be any outcome and price that the bidder can achieve by lying.

First, suppose that $x_i(o')<x_i(o)$, i.e., bidder $i$ is lying to achieve an outcome with a lower value. Then we know from individual rationality that
$p_i\leq\frac{v_i}{\gamma+1}$,
and we know from no-positive-transfers that $p_i'\geq0$. Thus, by the conditions of the theorem, we know that
\[(v_i-p_i)x_i(o)\geq v_i\frac{\gamma}{\gamma+1}x_i(o)>v_ix_i(o')\geq (v_i-p_i')x_i(o')\enspace.\]
In words, a profit-maximizing bidder prefers truthful reporting at the maximum individually-rational price to lying, even if $o'$ happened for free. Since bidder $i$ has super-profit-maximizing preferences and $v_i(o)>v_i(o')$, it follows $(o,p_i)\succ_i(o',p_i')$.

Next, suppose that $x_i(o')=x_i(o)$. We know that the mechanism is truthful for value maximizers, so it must be that $p_i\leq p_i'$, otherwise a value maximizer would lie to achieve $o'$, $p_i'$. Since any super-quasilinear bider will prefer a lower price for the same allocation, we can conclude that $(o,p_i)\succeq_i(o', p_i')$.

Finally, suppose that $x_i(o')>x_i(o)$. By individual rationality we know $(o,p_i)\succeq_i(0,0)$. Since the auction is truthful for value maximizers, we know that $p_i'>v_i$, otherwise a value maximizer would lie to achieve $(o',p_i')$. By definition of a super-quasilinear bidder, $p_i'>v_i$ implies $(0,0)\succ_i(o',p_i')$ and so we get $(o,p_i)\succeq_i(0,0)\succ_i(o',p_i')$, implying bidder $i$ will prefer to tell the truth.
\qed
\end{proof}

\subsection{Empirical Robustness}
To empirically evaluate Theorem~\ref{thm:robust}, we look at Yahoo marketplace data and ask {\em for what ROI constraint can we safely conclude that any super-quasilinear bidder would behave like a value maximizer?}

We employ the standard
separable click-through-rate (CTR) framework used to study sponsored search
auctions. In this model, $n$ ads compete for $m$ slots. Each ad $i$ has a private value per click $v_i$, and there exist public parameters
$(\alpha_1,\ldots,\alpha_m)$ and $(\beta_1,\ldots,\beta_n)$ such that, when ad
$i$ is shown in slot $j$, the user clicks on it with probability:
\[\Pr[\mbox{click on ad $i$ when shown in slot $j$}]=\alpha_j\beta_i, \]
Note that for any ad $i$ and slot $j$, we have $\frac{\Pr[\mbox{click on ad $i$ when shown in slot $j$}]}{\Pr[\mbox{click on ad $i$ when shown in slot $j+1$}]}=\frac{\alpha_j}{\alpha_{j+1}}$.

If we directly apply our theorem to this separable model we get the following conservative corollary stating that when the probability of a click drops substantially from slot $j$ to slot $j+1$, then a mild ROI constraint makes agents look like value maximizers:
\begin{corollary}[of Theorem~\ref{thm:robust}]
If $\frac{\alpha_j}{\alpha_{j+1}}\geq\frac{\gamma}{\gamma+1}$ for all $i$ in a standard sponsored search auction, then no bidder with super-quasi-linear preferences and a ROI constraint of $\gamma$ wishes to lie under GSP pricing.
\end{corollary}

This claim gives conditions under which no bidder has an incentive to lie,
independent of other bidders' bids; however, using the bids of other agents may
give sharper conditions. Given bids, we can ask a simpler question of whether
any bidder could possibly prefer to lie under current marketplace conditions:

\begin{lemma}\label{lem:native-robust}
When no two bidders have the same score $\beta_t$, and
\[\forall i,\; \frac{\alpha_i}{\alpha_i-\alpha_{i+1}}-\frac{\alpha_{i+1}}{\alpha_i-\alpha_{i+1}}\frac{\beta_{i+2}b_{i+2}}{\beta_{i+1}b_{i+1}}<\gamma\enspace,\]
no bidder with super-quasi-linear preferences and a ROI constraint of $\gamma$ wishes to lie under GSP pricing (i.e. under truthful pricing for value maximizers) holding others' bids fixed.
\end{lemma}

The proof (omitted) is similar to Theorem~\ref{thm:robust}.

We tested this lemma empirically by looking at bid data for the slot
auctions on Yahoo's homepage stream. We took a dataset consisting of over one
hundred thousand auctions from a brief period of time within a single day. 

Our results are striking --- if bidders require an ROI of 1, then 80\% of auctions would be such that no bidder can benefit by lying under GSP pricing. This strongly suggests that GSP may in fact be the appropriate auction for this setting. See Figure~\ref{fig:emp}.

\begin{figure}[h!]
\centering  
\hspace{-0mm}
\epsfxsize=5.5in \epsfbox{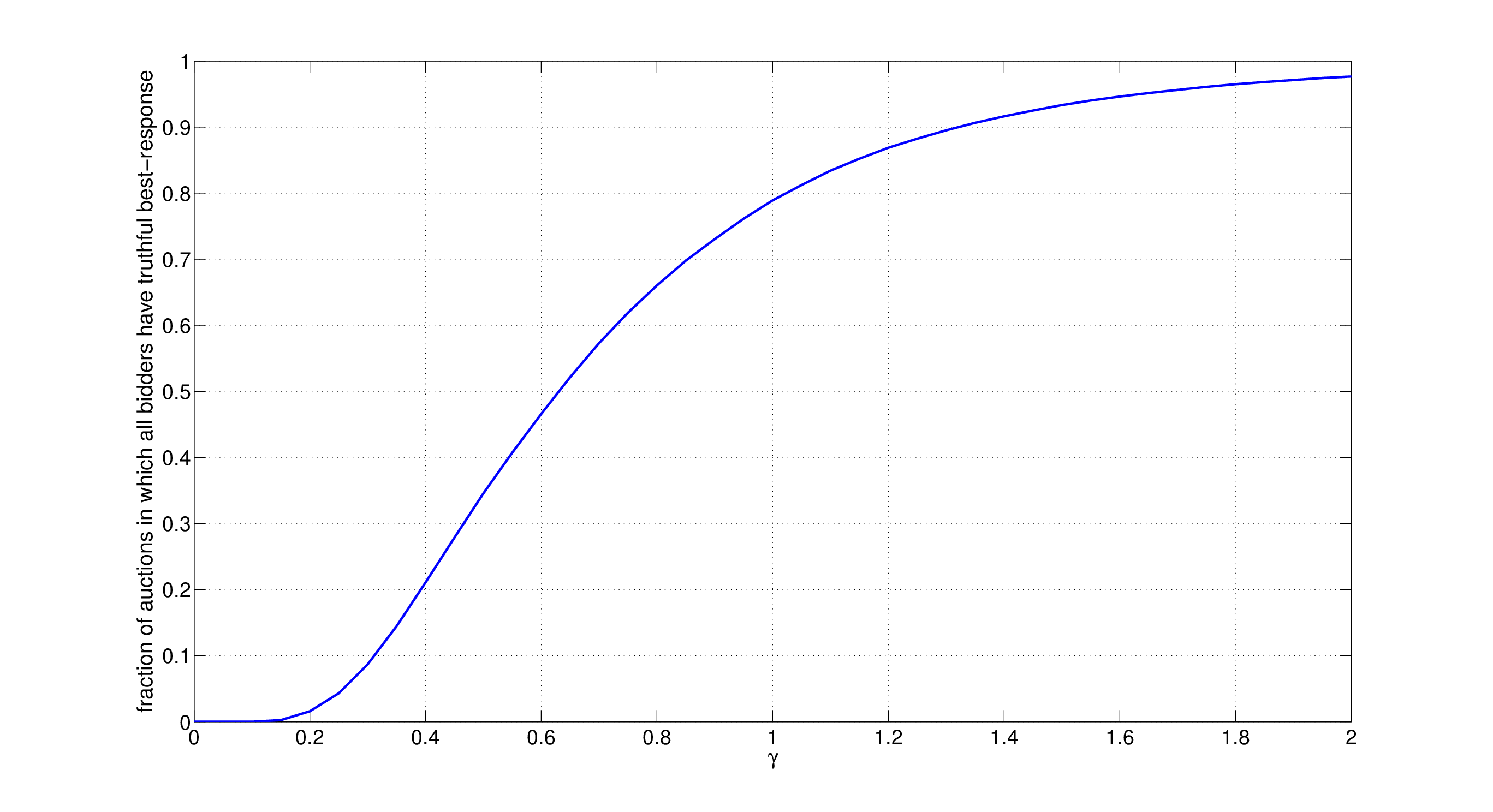}
\caption{\label{fig:emp} Illustration of the proportion of auctions in which truthtelling is a best-response for every advertiser, given the bids of all others, assuming all bidders have a ROI constraint of $\gamma$. At $\gamma=1$, 80\% of auctions are such that nobody should lie; as $\gamma$ approaches 2, virtually all auctions satisfy Lemma~\ref{lem:native-robust}. This is derived from a dataset of auctions from Yahoo's homepage stream, in which slot advertisements are interspersed in a stream of rich content links.}
\vspace{-0mm}
\end{figure}

%% file: value-max.tex
\section{The Greedy Mechanism for General Domains}  \label{sec:value-max}
\renewcommand{\O}{\mathcal{O}}
We now move beyond single-parameter preferences and introduce a family of
truthful mechanisms based on a greedy allocation algorithm. Analogous to the
Vickrey-Clarke-Groves auction, these greedy mechanisms are powerful because
they are truthful for value maximizers even under the most general valuation
models. 

\paragraph{The Greedy Auction.}
We first present a simple ``unweighted'' greedy auction for value maximizers with nonnegative values, defined in Algorithm~\ref{alg:greedy-mech}, that informally proceeds  as follows: 
\begin{enumerate}
\item Compute the outcome $o^*$ by the following algorithm:
\begin{enumerate}
\item Find the outcomes that maximize the value of the highest-value bidder.
\item Among the outcomes from (a), find the outcomes that maximize the value of the second-highest-value bidder.
\item Repeat until all bidders have been considered and call the outcome $o^*$; if more than one outcome remains, pick one arbitrarily.
\end{enumerate}
\item Payments are computed as the following ``externality'': among the (other) bidders whose values change when $i$ is present in the auction, identify the bidder who gets the most value if $i$ weren't present; charge this value to $i$ (in domains with well-ordered discrete typespaces, we charge a price equal to the next type above this critical value). 
\end{enumerate}

\begin{example}\label{ex:greedy}
Suppose there are three outcomes $\{o_1,o_2,o_3\}$ and four bidders with values $v_1=\{3,3,1\}$, $v_2=\{0.5,1,1\}$, $v_3=\{2,1,0\}$, and $v_4=\{0.5,0.5,0.5\}$.

The greedy maximizing auction for value maximizers chooses outcome $o_1$ as follows:
\begin{enumerate}
\item The value of the highest-value bidder is maximized by taking either outcome $o_1$ or $o_2$ since $v_1(o_1)=v_1(o_2)=3$.
\item Among outcomes $o_1$ and $o_2$, the value of the second-highest-value bidder is maximized by taking $o_1$ since $v_3(o_1)=2$.
\end{enumerate}
To compute prices, observe that the auction would still choose $o_1$ if bidder 1, 2, or 4 were removed from the auction, so these bidders pay $p_{\{1,2,4\}}=0$. For bidder 3, notice that the outcome would be $o_2$ if it were removed: neither bidder 1 nor bidder 4 care whether the outcome is $o_1$ or $o_2$, but bidder 2 gets $v_2(o_2)=1$ from $o_2$ instead of $v_2(o_1)=0.5$ from $o_1$. The price bidder 3 pays is bidder 2's value for $o_2$, i.e. $p_3=v_2(o_2)=1$.
\end{example}

\begin{algorithm}
	\SetAlgoNoLine
	\SetKwFunction{greedy}{Greedy}\SetKwFunction{greedymech}{GreedyMech}
	\SetKwProg{opt}{Algorithm}{}{}
	\SetKwProg{mech}{Mechanism}{}{}
	\opt{\greedy{$\{b_i\}$}}{
	\KwIn{Bids , one vector $b_i\in\Re^m$ per bidder.}
	\KwOut{An outcome $o^*$.}
	
	$C\leftarrow\O$\tcp*{Candidate outcomes}
	\For{$k=1$ to $n$}{
		$C\leftarrow\argmax_{o\in C}(\mbox{$k$-th highest bid for outcome $o$)}$\;
	}
	\Return{$o^*\in C$}\tcp*{Break ties according to any order on $\O$}
	}
	
	\mech{\greedymech{$\{b_I\}$}}{
		\KwIn{Bids $\{b_i\}$, one vector $b_i\in\Re^m$ per bidder.}
	\KwOut{An outcome $o^*$ and payments $\{p_i\}$.}
	
	$o^*\leftarrow\greedy(b)$\;
	\For{$i=1$ to $n$}{
		$o_{-i}^*\leftarrow\mbox{\greedy{$b_{-i},-\infty$}}$\;
		$p_i\leftarrow\max \left\{b_j(o_{-i}^*)\left|b_j(o_{-i}^*)>b_j(o^*)\right.\right\}$\;
	}
	\Return{$o^*$, $\{p_i\}$}\;
	
	}
	
	\caption{The Greedy Mechanism for Value Maximizers.}
\label{alg:greedy-mech}
\end{algorithm}

Our first theorem for this setting is that the greedy mechanism is truthful, as long as ties only happen in a zero-measure subset of values. A tie occurs only when two different (non-trivial) outcomes have exactly the same value for two different bidders. This condition is typically true if one chooses the Lebesgue measure or any probability density.

\begin{theorem}\label{thm:greedy-mech}
When values are nonnegative, the greedy auction for (simple) value maximizers is DSIC-AE as long as ties under truthful bidding happen with measure zero.
\end{theorem}
\begin{proof}{Proof of Theorem~\ref{thm:greedy-mech}.} Observe that the greedy mechanism proceeds in rounds. In each round, the bidder who can attain the highest value becomes the ``winner'' (unless there is a tie) and we discard outcomes that do not achieve this optimal value. Let $\mathcal O_r$ denote the eligible outcomes at round $r$, $i^*_r$ denote the bidder who wins round $r$ and $v_r$ denote the value it receives.

Say that a value vector $v$ results in a tie if there are distinct bidders $i$ and $j$ and outcomes $o_i$ and $o_j$ such that $v_i(o_i)=v_j(o_j)>0$. For now, assume that ties do not happen. Note the following:

\begin{itemize}
\item[(a)] In the round where bidder $i$ is the winner, she gets the outcome with the highest value to her---which is also the outcome with the highest value---among all remaining outcomes.
\item[(b)] In order to win in an earlier round $r$, bidder $i$ must express a value at least $c=v_r>\max_{o\in\mathcal O_r}v_i(o)$ and will pay $p_i=c$. (The inequality is strict because we assumed no ties.)
\end{itemize}

(a) implies that bidder $i$ is getting her favorite outcome among the outcomes remaining when she wins, so she cannot gain by changing the outcome in that round or a later round. (b) implies that bidder $i$ cannot gain by winning an earlier round, because she would need to pay a price greater than her maximum value for any remaining outcome.

Our previous argument assumed no ties occur. When the bid space is continuous, ties happen with measure zero and the mechanism is DSIC-AE.
\qed
\end{proof}

\paragraph{Transformed Greedy Mechanisms.} Next, we show how the greedy mechanism can be generalized. Algorithm~\ref{alg:greedy-tmech} defines the mechanism, which allows negative values, discretized bid spaces, arbitrary monotone transformations of bidders' values, and outcome-specific ``offsets.'' Informally, the {\em transformed greedy auction for value maximizers} transforms agent $i$'s value according to a monotone function $\theta_i:\Re\rightarrow\Re$ and greedily chooses an outcome according to $\theta$. Again, any mechanism in this family is truthful for any transformation of value:

\begin{algorithm}[h]
	\SetAlgoNoLine
	\SetKwFunction{tfgreedy}{TransformedGreedy}\SetKwFunction{ftgreedymech}{TransformedGreedyMech}
	\SetKwProg{opt}{Algorithm}{}{}
	\SetKwProg{mech}{Mechanism}{}{}
	\opt{\tfgreedy{$\{b_i\}$}}{
	\KwIn{Bids, one vector $b_i\in\Re^m$ per bidder.}
	\KwOut{An outcome $o^*$.}
	
	\For{$o\in\mathcal O$}{
		$z(o) \leftarrow\;$ \texttt{Sort} $\{ \theta_1(v_1(o)),\dots,\theta_n(v_n(o)),\theta_0(o) \}$ by $|\theta_i(v_i(o))|$ in descending order.
	}
	$o^*\leftarrow \;$ $o$ with greatest $z(o)$ according to a lexicographic order (breaking ties by order on $\O$)

\Return{$o^*$}\;
	}
	
	\mech{\ftgreedymech{$\{b_I\}$}}{
		\KwIn{Bids $\{b_i\}$, one vector $b_i\in\Re^m$ per bidder.}
	\KwOut{An outcome $o^*$ and payments $\{p_i\}$.}
	
	$o^*\leftarrow\tfgreedy{b}$\;
	\For{$i=1$ to $n$}{
		$p_i \leftarrow \inf \Big\{ b_i'(o^+) \,\big|\, o^+ = b_{-i},b_i' \text{ and } b_j(o^+) = b_j(o^*) \Big\}$
	}
	\Return{$o^*$, $\{p_i\}$}\;
	
	}
	
	\caption{The Transformed Greedy Family of Mechanisms.}
\label{alg:greedy-tmech}
\end{algorithm}

\begin{theorem}\label{thm:twm}
The transformed greedy auction is truthful for any nondecreasing transformations $\theta_i$ and offsets $\theta_0(o)$ as long as ties under truthful bidding happen with measure zero or the bid space is discretized with a minimum increment.
\end{theorem}


\begin{proof} {Proof of Theorem~\ref{thm:twm}.}
We prove the theorem when values are nonnegative. The reasoning for negative values is analogous.

Say that a value vector $v$ results in a tie if there are distinct $i,j\in\{0,1,\dots,n\}$ and outcomes $o_i$ and $o_j$ such that $\theta_i(o_i)=\theta_j(o_j)>0$. For now, assume that ties do not happen.

Note the following:
\begin{itemize}
\item[(a)] The mechanism is IR by construction, so a bidder will never wish to lie to achieve a smaller value.
\item[(b)] By definition the mechanism charges the same price to all types that get the same value, so there is no incentive to misreport to get the same outcome at a lower price.
\item[(c)] Let $o^*$ be the outcome chosen by the mechanism and $z^*=z(o)$. Let $\mathcal O_r$ denote the set of outcomes that have the same value as $z^*$ in positions $1,\dots r-1$, and suppose that bidder $i$ ranks $r$-th in the vector $z^*$. Then for all $k<r$, $z_k^*>\max_{o\in\mathcal O_k}\theta_i(v_i(o))$ (strictly since we assumed no ties).
\end{itemize}
(a) and (b) imply that if there is a reason to lie about one's type, it is to achieve a strictly higher value. In order to do this, a bidder $i$ must move up in the ranking; however, this requires placing a bid $b_i(o)$ for some outcome such that $\theta_i(b_i(o))\geq z_{r-1}^*$ and will cause bidder $i$ to pay $p_i\geq \theta_i^{-1}(z_{r-1}^*)$. By (c) this implies $p_i>\max_ov_i(o)$, i.e. bidder $i$ would necessarily pay more than its value.

As in~\ref{thm:greedy-mech}, ties happen with measure zero when the bid space is continuous and we find that the mechanism is DSIC-AE.

Moreover, when values are discretized additional analysis reveals that the mechanism is DSIC. Note that our argument above only gets into trouble when the inequality in (c) is in fact equality, i.e. $z_{r-1}^*=\max_{o\in\mathcal O_k}\theta_i(v_i(o))$. In particular, this implies that there is an outcome $o'\neq o^*$ such that $z_{r-1}^*=\theta_i(v_i(o'))$. Suppose that $i$ lies and reports $v_i'$ instead, inducing the outcome $o'$. The payment $p_i'$ will be (replacing $\inf$ with $\min$ because we are in a well-ordered space with a minimum bid increment):
\[p_i'=\min \left\{b_i''(o^+)\left|o^+=\mbox{\tfgreedy{$b_{-i},b_i''$} and }b_j'(o^+)=b_j'(o')\right.\right\}\geq v_i(o')+\delta\]
where $\delta$ represents the minimum bid increment and the inequality follows because $i$ needed to increase its bid to get $o'$. This implies $p_i'\geq v_i(o')+\delta>v_i(o')$ and therefore $i$ does not benefit by lying in this scenario either. Thus, the mechanism is DSIC when values are discrete.
\qed
\end{proof}

\paragraph{Bounding Total Value.} Although the existence of truthful mechanisms is a positive result, it may be unfortunate for a mechanism designer because a greedy optimization cannot generally guarantee an approximation better than $\Theta(n)$ to the optimal total generated value, i.e., sum of agents' values under an outcome.
\begin{theorem}
\label{thm:tgreedy-lower}
No transformed greedy mechanism can guarantee more than than a $\frac1{n}$ fraction of the optimal total generated value in the worst case.
\end{theorem}
\begin{proof}{Proof of Theorem \ref{thm:tgreedy-lower}}
Suppose there are at least two outcomes and $n$ bidders. Pick any non-trivial transformed greedy mechanism (i.e. a mechanism that realizes at least two outcomes for some inputs) with transformations $\theta_i$ and offsets $\theta_0$.

We label bidders and choose $z$ such that $\theta_1(z)>\theta_i(z-\epsilon)$ for all $i>1$. Let bidder 1 have value $v_1=(z, z-\epsilon, 0,\dots)$ and let other bidders have values $v_{i}=(0,z-\epsilon,0,\dots)$. The value-optimal outcome is outcome 2, with a total value of $(z-\epsilon)n$; however, a greedy mechanism must choose outcome 1 because $\theta_1(z)>\theta_i(z-\epsilon)$ for all $i$, leading to a total value of $1$. This implies a worst-case total value approximation of $\frac{z}{n(z-\epsilon)}$.

The non-transformed greedy mechanism will achieve a total generated value of at least $\max_{i,o}v_i(o)\geq\frac1n\max_o\sum_iv_i(o)$ and therefore guarantee at least a $\frac1n$ fraction of the optimal total value.
\qed
\end{proof}

%% file: conclusion.tex
\section{Conclusion}  \label{sec:conclusion}

The standard philosophy in theory work is that no model will be perfect, but careful analysis of a simple model can generate insights. In mechanism design, it is typically taken for granted that a simple quasilinear utility is the natural starting point for modeling any selfish agent. However, motivated by an abundance of anecdotes and experience with bidding behavior in industry, we suggest that an alternative model in which agents maximize value with minimal price sensitivity is also important. While our model behaves poorly in certain extremes, it is a simple model that captures a fundamental type of behavior.

Our work addresses a handful of basic questions that are solved for quasilinear bidders in the classic theory, but many more remain open, for example:
\begin{itemize}
\item {\em Can we prove a Roberts's theorem for value maximizers?}
\item {\em How do bidders' preferences aggregate across auctions?} 
\item {\em What are bidders' preferences over lotteries?} 

\item {\em What happens in non-truthful mechanisms? Does revenue equivalence hold?}
\end{itemize}

We expect that many of these questions will have answers that suggest
interesting and important twists on the standard theory. Follow up work may
also find empirical reasons to refine the value maximizer preference structure
we presented and analyzed here.